\documentclass[notitlepage,a4,10.5pt]{article}

\usepackage{graphicx}

\usepackage{color}
\usepackage{amsmath}%
\usepackage{amsfonts}%
\usepackage{amssymb}%
\usepackage{mathrsfs}
\usepackage{amsthm}
\usepackage{float}

\usepackage{authblk}

\usepackage[left=2.5cm,right=2.5cm,top=3cm,bottom=3cm]{geometry}

\newcommand{\mc}{\mathcal}
\newcommand{\cp}{\times}

\newcommand{\bol}{\boldsymbol}

\newcommand{\w}{\wedge}
\newcommand{\lr}[1]{\left({#1}\right)}
\newcommand{\lrs}[1]{\left[{#1}\right]}
\newcommand{\lrc}[1]{\left\{{#1}\right\}}

\newcommand{\mf}{\mathfrak}
\newcommand{\p}{\partial}

\newcommand{\ti}[1]{\textit{#1}}
\newcommand{\tb}[1]{\textbf{#1}}

\newtheorem{theorem}{\textit{Theorem}}
\newtheorem{proposition}{\textit{Proposition}}

\begin{document}

\title{Generalization of Hamiltonian Mechanics to a Three-Dimensional Phase Space}
\author[1]{Naoki Sato} 
\affil[1]{Graduate School of Frontier Sciences, \protect\\ The University of Tokyo, Kashiwa, Chiba 277-8561, Japan \protect\\ Email: sato\_naoki@edu.k.u-tokyo.ac.jp}
\date{\today}
\setcounter{Maxaffil}{0}
\renewcommand\Affilfont{\itshape\small}

    \maketitle
    \begin{abstract}
		Classical Hamiltonian mechanics is realized by the action of a Poisson bracket 
		on a Hamiltonian function. The Hamiltonian function is a constant of motion (the energy) of the system. 
		The properties of the Poisson bracket are encapsulated in the symplectic $2$-form, a closed
		second order differential form. Due to closure, the symplectic $2$-form is preserved by the Hamiltonian flow,  
		and it assigns an invariant (Liouville) measure on the phase space through the Lie-Darboux theorem.  
		In this paper we propose a generalization of classical Hamiltonian mechanics to a three-dimensional phase space:
		the classical Poisson bracket is replaced with a generalized Poisson bracket acting on a pair of Hamiltonian functions, 
		while the symplectic $2$-form is replaced by a symplectic $3$-form. 
		We show that, using the closure of the symplectic $3$-form, a result analogous to the classical Lie-Darboux theorem holds: 
		locally, there exist smooth coordinates such that the components of the symplectic $3$-form are constants, and the
		phase space is endowed with a preserved volume element. 
		Furthermore, as in the classical theory, the Jacobi identity for the generalized Poisson bracket 
		mathematically expresses the closure of the associated symplectic form. 
		As a consequence, constant skew-symmetric third order contravariant tensors always define generalized Poisson brackets. 
		This is in contrast with generalizations of Hamiltonian mechanics postulating the  
	  fundamental identity as replacement for the Jacobi identity. 
		In particular, we find that the fundamental identity represents a stronger requirement than 
		the closure of the symplectic $3$-form. 
		\end{abstract}

\section{Introduction}
In the bracket formalism \cite{Morrison80,Morrison82,Olv}, 
a classical Hamiltonian system 
is defined as a pair 
consisting of a smooth manifold $\mc{M}$ of dimension $n$ and a Poisson bracket
$\left\{\circ,\circ\right\}:C^{\infty}\lr{\mc{M}}\times C^{\infty}\lr{\mc{M}}\rightarrow C^{\infty}\lr{\mc{M}}$ on the ring of real valued smooth functions $C^{\infty}\lr{\mc{M}}$ on $\mc{M}$. The Poisson bracket satisfies the following axioms:
\begin{subequations}\label{axiomsPB}
\begin{align}
&\left\{aF+bG,H\right\}=a\left\{F,H\right\}+b\left\{G,H\right\},~~~~\left\{H,aF+bG\right\}=a\left\{H,F\right\}+b\left\{H,G\right\},\label{bilinearity}\\
&\left\{F,F\right\}=0,\label{alternativity}\\
&\left\{F,G\right\}=-\left\{G,F\right\},\label{antisymmetry}\\
&\left\{FG,H\right\}=F\left\{G,H\right\}+\left\{F,H\right\}G,\label{Leibniz}\\
&\left\{F,\left\{G,H\right\}\right\}+\left\{G,\left\{H,F\right\}\right\}+\left\{H,\left\{F,G\right\}\right\}=0,\label{JI}
\end{align}
\end{subequations}
for all $a,b\in\mathbb{R}$ and $F,G,H\in C^{\infty}\lr{\mc{M}}$. 
Equations \eqref{bilinearity}, \eqref{alternativity}, \eqref{antisymmetry},
\eqref{Leibniz}, and \eqref{JI} express bilinearity, alternativity, 
antisymmetry, Leibniz rule, and Jacobi identity respectively. 
Bilinearity \eqref{bilinearity} guarantees that the Poisson bracket defines an algebra over $\mathbb{R}$. Given an observable $F\in C^{\infty}\lr{\mc{M}}$ and an Hamiltonian (energy) $H\in C^{\infty}\lr{\mc{M}}$, the time evolution of $F$ is expressed as
\begin{equation}
\frac{dF}{dt}=\left\{F,H\right\}=F_i\mc{J}^{ij}H_j.
\end{equation}
In this notation $\mc{J}^{ij}$ is a skew-symmetric second order contravariant tensor (the Poisson operator associated with the Poisson bracket), summation on repeated indexes is used, and lower indexes applied to a function specify partial derivatives, e.g. $F_i=\p F/\p x^i$ with $x^i$ the $i$th coordinate on $\mc{M}$, $i=1,...,n$.   
Then, alternativity \eqref{alternativity} ensures that energy is conserved, since $dH/dt=\left\{H,H\right\}=0$. Antisymmetry \eqref{antisymmetry} follows from the axioms \eqref{bilinearity} and \eqref{alternativity}. 
The Leibniz or derivation rule \eqref{Leibniz} implies that the Poisson bracket
behaves as a derivation since $d\lr{FG}/dt=F\lr{dG/dt}+\lr{dF/dt}G$.
Finally, the Jacobi identity \eqref{JI} is equivalent to the existence
of a closed $2$-form of even rank $2m=n-s$, the symplectic $2$-form $\omega$. 
Here, $s$ is the dimension of the kernel of the Poisson bracket. 
Due to the Lie-Darboux theorem \cite{Littlejohn82,Arnold89,DeLeon89}, the closure of the $2$-form $\omega$  
further implies that the phase space is locally spanned by $2m=n-s$ canonically conjugated variables $\lr{p^i,q^i}$, $i=1,...,m$, and $s=n-2m$ Casimir invariants $C^i$, $i=1,...,s$, which fill the center (kernel) of the Poisson bracket:
\begin{equation}
\frac{dC^i}{dt}=\left\{C^i,H\right\}=0~~~~\forall H\in C^{\infty}\lr{\mc{M}},~~~~i=1,...,s.
\end{equation}
The naming Casimir invariant used in this context originates from the Lie algebra associated with angular momentum \cite{Morrison98}. 
The equations of motion therefore take the local canonical form
\begin{equation}
\frac{dp^i}{dt}=-\frac{\p H}{\p q^i},~~~~\frac{dq^i}{dt}=\frac{\p H}{\p p^i},~~~~\frac{dC^j}{dt}=0,~~~~i=1,...,m,~~j=1,...,s.\label{HCEq}
\end{equation}

In his 1973 paper \cite{Nambu}, Nambu proposed 
what he calls a `possible generalization of classical Hamiltonian dynamics to a three-dimensional phase space'. 
Here, a three-dimensional phase space means that the classical canonical pair $\lr{p,q}$ is replaced by a canonical triplet $\lr{p,q,r}$, while the number of generating functions (Hamiltonians) is increased to two. 
Then, for the basic $n=3$ case with Hamiltonians $G$ and $H$ , 
in place of Hamilton's canonical equations, Nambu's canonical equations are introduced as follows:
\begin{equation}
\frac{dp}{dt}=\frac{\p G}{\p q}\frac{\p H}{\p r}-\frac{\p G}{\p r}\frac{\p H}{\p q},~~~~\frac{dq}{dt}=\frac{\p G}{\p r}\frac{\p H}{\p p}-\frac{\p G}{\p p}\frac{\p H}{\p r},~~~~\frac{dr}{dt}=\frac{\p G}{\p p}\frac{\p H}{\p q}-\frac{\p G}{\p q}\frac{\p H}{\p p}.\label{NC}
\end{equation}
Setting $\lr{x^1,x^2,x^3}=\lr{p,q,r}$, system \eqref{NC} can be written through a ternary operation (Nambu bracket) 
$\left\{\circ,\circ,\circ\right\}:C^{\infty}\lr{\mc{M}}\cp C^{\infty}\lr{\mc{M}}\cp C^{\infty}\lr{\mc{M}}\rightarrow C^{\infty}\lr{\mc{M}}$ as below:
\begin{equation}
\frac{dx^i}{dt}=\left\{x^i,G,H\right\}=\epsilon^{ijk}G_jH_k,~~~~i=1,2,3.
\end{equation}
In this equation, $\epsilon^{ijk}$ denotes the three-dimensional Levi-Civita symbol.
More generally, the evolution of an observable $F\in C^{\infty}\lr{\mc{M}}$ takes the form
\begin{equation}
\frac{dF}{dt}=\left\{F,G,H\right\}=\epsilon^{ijk}F_iG_jH_k.\label{NB}
\end{equation}
Notice that due to the skew-symmetry of $\epsilon^{ijk}$ both $G$ and $H$ are constant in time. 

The generalization of the Nambu bracket occurring in equation \eqref{NB} 
to an algebraic framework analogous to the Poisson bracket of classical Hamiltonian mechanics 
has proven difficult, especially in the context of quantization, because 
the generalization of the Poisson bracket axioms \eqref{axiomsPB} (in particular, the generalization of the Jacobi identity \eqref{JI}) to the Nambu bracket is nontrivial \cite{Ho}.
Historically, Nambu dynamics was first placed in a Lie algebraic setting in \cite{Bial91}.  
Here, the authors considered the properties of the bracket $\left\{\circ,\circ\right\}_G=\left\{\circ,G,\circ\right\}$ obtained by
fixing one of the Hamiltonians, and showed that $\left\{\circ,\circ\right\}_G$ satisfies the Jacobi identity of Poisson brackets 
when the Nambu bracket is defined through structure constants of Lie-Poisson brackets and $G$ is the corresponding Casimir invariant (note that this construction is not limited to the Lie algebra of the rotation group chosen by Nambu). This problem
is intimately related to the generalization of the Jacobi identity to Nambu mechanics because 
it is desirable that a triple bracket with an equivalent Poisson structure satisfies the generalized identity. 
In \cite{Bial91}, the authors also derived a field theoretic (infinite-dimensional) Nambu bracket though the Lie
algebra associated with the Weyl-Wigner representation (these aspects are also reviewed in \cite{Morrison98,Bloch13}). 
Several authors \cite{Ho,Tak} have proposed the following set of axioms for the Nambu bracket: 
trilinearity, skew-symmetry, Leibniz rule, and fundamental identity. Respectively, they are written as
\begin{subequations}\label{NAx}
\begin{align}
\left\{a F_1+b F_2,F_3,F_4\right\}=&a\left\{F_1,F_3,F_4\right\}+b\left\{F_2,F_3,F_4\right\},\label{tri}\\
\left\{F_1,F_2,F_3\right\}=&\epsilon^{ijk}\left\{F_i,F_j,F_k\right\},~~~~i,j,k=1,2,3~~~~({\rm not~summed}),\label{ssn}\\
\left\{F_1F_2,F_3,F_4\right\}=&F_1\left\{F_2,F_3,F_4\right\}+F_2\left\{F_1,F_3,F_4\right\},\label{LeibN}\\
\left\{\left\{F_1,F_2,F_3\right\},F_4,F_5\right\}=&\left\{\left\{F_1,F_4,F_5\right\},F_2,F_3\right\}
+\left\{F_1,\left\{F_2,F_4,F_5\right\},F_3\right\}+\left\{F_1,F_2,\left\{F_3,F_4,F_5\right\}\right\},\label{FI}
\end{align}
\end{subequations}   
for all $F_1,F_2,F_3,F_4,F_5\in C^{\infty}\lr{\mc{M}}$ and $a,b\in\mathbb{R}$. 
Observe that in \eqref{tri} the linearity condition on the second and third argument has been omitted.
As in the Poisson bracket case, linearity in the other arguments follows from linearity in the first argument and skew-symmetry. 
Alternativity has been absorbed in the skew-symmetry axiom \eqref{ssn}.
In the Leibinz (derivation) rule \eqref{LeibN}, the convention is adopted that $F_1$ and $F_2$ both appear before the Nambu bracket on the right-hand side of the equation.
Finally, the fundamental identity \eqref{FI}, which replaces the Jacobi identity \eqref{JI} for the Poisson bracket, 
implies distribution of derivatives (a property satisfied by Poisson brackets), i.e. given two Hamiltonians $F_4$ and $F_5$ one has
\begin{equation}
\frac{d}{dt}\left\{F_1,F_2,F_3\right\}=\left\{\frac{dF_1}{dt},F_2,F_3\right\}+\left\{F_1,\frac{dF_2}{dt},F_3\right\}+\left\{F_1,F_2,\frac{dF_3}{dt}\right\}.
\end{equation}
The fundamental identity \eqref{FI} leads to the property that 
the bracket $\left\{\circ,\circ\right\}_G=\left\{\circ,G,\circ\right\}$ defined by fixing 
the second entry with a given generating function $G$ assigns a Poisson algebra.
This can be verified by observing that equation \eqref{FI} reduces to \eqref{JI} when $F_1=F_4=G$. 
However, the fundamental identity \eqref{FI} also leads to the result that constant skew-symmetric 3-tensors 
do not define a Nambu bracket in general, a situation pointing to the fact that the axiom \eqref{FI} 
is more stringent than the Jacobi identity \eqref{JI} required for a Poisson bracket (on this point, see \cite{Tak}).

As discussed in \cite{Ho,Tak,Cha}, the axioms \eqref{NAx} can be further generalized to the $N$-ary case ($N\leq n$) in which
the evolution of an observable $F\in C^{\infty}\lr{\mc{M}}$ in $n$-dimensional phase space is determined by $N-1$ Hamiltonians, $H_1,...,H_{N-1}\in C^{\infty}\lr{\mc{M}}$ as 
\begin{equation}
\frac{dF}{dt}=\left\{F,H_1,...,H_{N-1}\right\}.
\end{equation}
The main case discussed by Nambu corresponds to $n=N=3$. 
More generally, Nambu dynamics with $n=N$ is associated with integrable systems endowed with $n-1$ constants of motion.  

In this paper, we wish to address the following question: `How can we generalize Hamiltonian mechanics to allow triplets (and possibly $N$-tuples, $N\geq 2$) of variables as the building blocks for the phase space?'  
To achieve this goal, we shall not consider the problem of quantization (see \cite{Awata,Dito} on the difficulties encountered in the quantization of Nambu brackets defined by \eqref{NAx}), 
but instead focus on the minimal differential-geometric structure that should be possessed by a generalized phase space, 
a closed symplectic form of degree $N$. 
We find that, for $N>2$, the mathematical identity (Jacobi identity) expressing the closure of the symplectic form ceases to coincide with the fundamental identity (the bracket identity \eqref{FI} proposed as generalization of the Jacobi identity \eqref{JI} for the classical Poisson bracket). Hence, the fundamental identity represents a stronger condition than what is required for a system
to be Hamiltonian in the geometric sense above, and the fundamental identity \eqref{FI} is replaced with a weaker condition
expressing the closure of the relevant symplectic $N$-form. 
Here, the bracket formalism is not essential to the definition of the generalized Hamiltonian framework, 
and all bracket axioms are deducible from the underlying symplectic structure of the phase space.

Although the construction of generalized Hamiltonian mechanics presented here does not come with an immediate pathway to quantization, 
it possesses several merits: any constant skew-symmetric 3-tensor always satisfies the Jacobi identity
and therefore represents a generalized Poisson operator, and  a theorem, analogous to the classical Lie-Darboux theorem
for closed $2$-forms, holds, implying that the phase space has a local set of coordinates that define an invariant (Liouville) measure. 
We also take the view that the applicability of any generalization of Hamiltonian mechanics 
should not be limited to integrable dynamical systems with $n=N$. Instead, it should include systems with arbitrary dimensionality $n\geq N$ and a desired number of Hamiltonians $H_1,...,H_{N-1}$, $N-1\geq 1$. 
The theory discussed here applies to $n$-dimensional systems ($n\geq 3$) with $N-1=2$ Hamiltonians, and 
the formalism has a potential extension to the general case with $n\geq N$ dimensions and $N-1\geq 1$ Hamiltonians. 

One may wonder whether such generalizations of Hamiltonian mechanics 
represent redundant formulations of the classical theory, or 
whether they bring physical insight into practical problems.  
The relationship between classical Hamiltonian mechanics and  
the Nambu bracket defined through the axioms \eqref{NAx} has been studied by 
several authors \cite{Bayen,Yon,Hor1,Hor2},  
who have highlighted the range of interchangeability of the formalisms 
and suggested applications to statistical mechanics, where the classical Liouville
measure is replaced by the invariant measure associated with the volume element of Nambu's phase space 
(an invariant measure is needed for the definition of entropy and the applicability of the
ergodic ansatz \cite{StatMech,Ergodic}), quantum mechanics, and fluid mechanics \cite{Blender,Nevir}. 
As it will be shown in the following sections, the geometric generalization of Hamiltonian mechanics presented here 
has the property that, 
exception made for the three-dimensional case $n=3$, 
the bracket obtained by fixing one of the constants, say $G$, according to
$\left\{\circ,\circ\right\}_G=\left\{\circ,G,\circ\right\}$
fails to define a Poisson bracket. 
This fact holds true even when the generalized Poisson operator $\mc{J}^{ijk}$ generating the dynamics as $dF/dt=\left\{F,G,H\right\}=\mc{J}^{ijk}F_iG_jH_k$
can be transformed to a tensor whose components are given by Levi-Civita symbols (one for each canonical triplet).
This suggests that there may be systems that possess a generalized Hamiltonian structure, but that
do not exhibit a classical phase space. This fact would make the theory indispensable 
if such systems were found in the natural world.  

The present paper is organized as follows.
In section 2 we review the geometric formulation of classical Hamiltonian mechanics.
In section 3, based on the same geometric construction, 
we define a generalization of Hamiltonian mechanics to dynamical systems possessing two Hamiltonians $G$ and $H$ , i.e. $N-1=2$ and $n\geq 3$. 
In section 4 we discuss the relationship between the closure of symplectic forms and the Jacobi identity. 
In particular, we show that the fundamental identity \eqref{FI} appearing in the set of axioms \eqref{NAx} ceases to 
coincide with the closure of the symplectic form when $N\geq3$.
In section 5, we extend the classical Lie-Darboux theorem, 
and show that phase space is locally spanned by a set of coordinates defining an invariant measure. 
In these coordinates, the components of the symplectic $3$-form are constants.  
The conditions for the local existence of canonical triplets are also discussed. 
Concluding remarks are given in section 6.

Finally, we remark again that, for clarity of exposition, we are concerned only with the generalization of Hamiltonian mechanics 
to systems with two Hamiltonians. Nevertheless, systems with a higher number of Hamiltonians $N-1>2$ 
are expected to admit analogous constructions. 

\section{Classical Hamiltonian Mechanics}
For simplicity, we restrict our attention to Euclidean space, $\mc{M}=\mathbb{R}^n$. 
Consider a smooth bounded domain $\Omega\subset\mathbb{R}^n$. 
Let $\bol{x}$ denote the position vector in $\mathbb{R}^n$, $\lr{x^1,...,x^n}$ a coordinate system in $\Omega$, 
$\lr{\p_1,...,\p_n}$ the associated tangent basis, $T\Omega$ the tangent space to $\Omega$, and $T^\ast\Omega$ the cotangent space to $\Omega$. 

A bivector field $\mc{J}\in\bigwedge^2T\Omega$ with expression
\begin{equation}
\mc{J}=\sum_{i<j}\mc{J}^{ij}\p_i\w\p_j,
\end{equation}
where $\mc{J}^{ij}\in C^{\infty}\lr{\Omega}$, $i,j=1,...,n$, 
are the components of a skew-symmetric second order contravariant tensor, 
is called a Poisson operator provided that 
the Jacobi identity
\begin{equation}
\mc{J}^{im}\frac{\p\mc{J}^{jk}}{\p x^m}+\mc{J}^{jm}\frac{\p\mc{J}^{ki}}{\p x^{m}}+\mc{J}^{km}\frac{\p\mc{J}^{ij}}{\p x^{m}}=0,~~~~i,j,k=1,...,n,\label{JI2} 
\end{equation}
is satisfied. 
Using the Poisson bracket $\left\{F,G\right\}=F_i\mc{J}^{ij}G_j$, one can verify that
equation \eqref{JI2} is equivalent to \eqref{JI}.

Given a smooth vector field $X=X^i\p_i\in T\Omega$, equations of motion are assigned as follows:
\begin{equation}
\frac{d\bol{x}}{dt}=X.\label{EoM}
\end{equation}
The vector field $X$ defines a Hamiltonian system whenever there exist 
a smooth function (Hamiltonian) $H\in C^{\infty}\lr{\Omega}$ 
and a Poisson operator $\mc{J}$ such that
\begin{equation}
X^i=\mc{J}^{ij}H_j,~~~~i=1,...,n.\label{HEoM} 
\end{equation}
The central result stemming from the Jacobi identity \eqref{JI2} 
is the so-called Lie-Darboux theorem \cite{Littlejohn82,Arnold89,DeLeon89},
which can be stated as follows:
\begin{theorem}
Let $\mc{J}\in\bigwedge^2T\Omega$ denote a smooth bivector field of rank $2m$ on a smooth manifold $\Omega$ of dimension $n=2m+s$, $s\geq 0$. Assume that the Jacobi identity \eqref{JI2} holds. Then, for every point $\bol{x}\in\Omega$ 
there exist a neighborhood $U\subset\Omega$ of $\bol{x}$ 
and local coordinates $\lr{p^1,...,p^m,q^1,...,q^m,C^1,...,C^s}\in C^{\infty}\lr{U}$ with associated tangent basis 
$\lr{\p_{p^1},...,\p_{p^m},\p_{q^1},...,\p_{q^m},\p_{C^1},...,\p_{C^s}}$ such that
\begin{equation}
\mc{J}=\sum_{i=1}^m \p_{q^i}\w \p_{p^i}~~~~{\rm in}~~U.\label{LDThmJ}
\end{equation}   
\end{theorem}
\noindent Equation \eqref{LDThmJ} implies that, in $U$,
\begin{subequations}
\begin{align}
\frac{dp^i}{dt}=&i_{X}dp^i=\mc{J}\lr{dp^i,dH}=\sum_{k=1}^m\lr{\frac{\p p^i}{\p q^k}\frac{\p H}{\p p^k}-\frac{\p p^i}{\p p^k}\frac{\p H}{\p q^k}}=-\frac{\p H}{\p q^i},\\
\frac{dq^i}{dt}=&i_Xdq^i=\mc{J}\lr{dq^i,dH}=\sum_{k=1}^m\lr{\frac{\p q^i}{\p q^k}\frac{\p H}{\p p^k}-\frac{\p q^i}{\p p^k}\frac{\p H}{\p q^k}}=\frac{\p H}{\p p^i},\\
\frac{dC^j}{dt}=&i_XdC^j=\mc{J}\lr{dC^j,dH}=\sum_{k=1}^m\lr{\frac{\p C^j}{\p q^k}\frac{\p H}{\p p^k}-\frac{\p C^j}{\p p^k}\frac{\p H}{\p q^k}}=0,
\end{align}
\end{subequations} 
where $i=1,...,m$ and $j=1,...,s$. In this notation $i$ represents the contraction operator (interior product), and the bivector field $\mc{J}$ 
acts on arbitrary $1$-forms $dF$ and $dG$ according to $\mc{J}\lr{dF}=\mc{J}^{ij}F_j\p_i$ and $\mc{J}\lr{dG,dF}=\mc{J}^{ij}G_iF_j$.
Hence, the equations of motion \eqref{HEoM} can be locally written in the canonical form \eqref{HCEq}.
Notice that any differentiable function $f\lr{C^1,...,C^s}$
of the coordinates $C^j$, $j=1,...,s$, is preserved by the Hamiltonian system \eqref{HEoM} regardless
of the specific form of the Hamiltonian $H$. Indeed, 
\begin{equation}
\mc{J}\lr{dC^j}=\sum_{i=1}^m\lr{\frac{\p C^j}{\p p^i}\p_{q^i}-\frac{\p C^j}{\p q^i}\p_{p^i}}=0,~~~~j=1,...,s.
\end{equation} 
The coordinates $C^j$, $j=1,...,s$, are called Casimir invariants. 

Consider the $2m$-dimensional submanifold 
$\Omega_C\subset\Omega$ defined by $\Omega_C=\left\{\bol{x}\in\Omega~\vert~C^j=c^j, j=1,...,s\right\}$ 
with $c^j\in\mathbb{R}$, $j=1,...,s$. Let $H_C=H\lr{p^1,...,p^m,q^1,...,q^m,c^1,...,c^s}$
be the value of the Hamiltonian $H$ on $\Omega_C$. Then, in $\Omega_C$, the equations of motion \eqref{HEoM} can be cast in the form
\begin{equation}
i_{X}\omega=-dH_C,\label{CHM}
\end{equation}
where $\omega\in\bigwedge^2 T^\ast\Omega$ is the closed $2$-form with local expression
\begin{equation}
\omega=\sum_{i=1}^m dp^i\w dq^{i}.\label{LDThm2}
\end{equation} 
Equation \eqref{LDThm2} represents an alternative statement of the Lie-Darboux theorem. More precisely, one has:
\begin{theorem}
Given a smooth closed $2$-form $\omega$ of rank $2m$ on a smooth manifold $\Omega$ of dimension $n=2m+s$, $s\geq0$,  
for every point $\bol{x}\in\Omega$ there exist a neighborhood $U\subset\Omega$ of $\bol{x}$ 
and local coordinates $\lr{p^1,...,p^m,q^1,...,q^m,C^1,...,C^s}\in C^{\infty}\lr{U}$ such that
\begin{equation}
\omega=\sum_{i=1}^m dp^i\w dq^{i}~~~~{\rm in}~~U.\label{LDThm}
\end{equation}   
\end{theorem}
\noindent The closure of the $2$-form $\omega$, which is called symplectic $2$-form, is expressed by the condition $d\omega=0$. 
In the coordinate system $\lr{x^1,...,x^n}$, the $2$-form $\omega$ has expression
\begin{equation}
\omega=\sum_{i<j}\omega_{ij}dx^i\w dx^j,
\end{equation}
where $\omega_{ij}$ is a skew-symmetric second order covariant tensor, $\omega_{ij}=-\omega_{ji}$, $i,j=1,...,n$. 
Then, the closure condition can be written as 
\begin{equation}
d\omega=\sum_{i<j<k}\lr{\frac{\p\omega_{ij}}{\p x^k}+\frac{\p\omega_{jk}}{\p x^i}+\frac{\p\omega_{ki}}{\p x^j}}dx^i\w dx^j\w dx^k=0.\label{closure}
\end{equation} 
Recalling Cartan's homotopy formula for the Lie derivative of a differential form,   
\begin{equation}
\mf{L}_{X}\omega=\lr{di_X+i_Xd}\omega, 
\end{equation}
and using the closure of $\omega$, 
equation \eqref{CHM} implies that $\omega$ is invariant along the vector field $X$, 
\begin{equation}
\mf{L}_{X}\omega=0.
\end{equation}
In addition to the conservation of the symplectic $2$-form $\omega$, Hamilton's canonical equation \eqref{HCEq} imply that
the phase space (Liouville) measure
\begin{equation}
d\Pi=dp^1\w...\w dp^m\w dq^1\w...\w dq^m\w dC^1\w...\w dC^s,
\end{equation} 
is preserved by the flow $X$, i.e.
\begin{equation}
\mf{L}_{X}d\Pi=0.
\end{equation}
The phase space measure $d\Pi$ is the cornerstone of classical statistical mechanics. 

To summarize, there are two dual geometrical formulations of classical Hamiltonian mechanics, 
one based on the Poisson operator $\mc{J}$, the other centered on the symplectic $2$-form $\omega$. 
The first formulation \eqref{HEoM} has the merit that the role of the kernel of $\mc{J}$, which is spanned by the Casimir invariants, 
is made explicit, while to define a Hamiltonian system by \eqref{CHM} one must require that $X\perp{\rm ker}\lr{\omega}$ 
and that the Hamiltonian is independent of the Casimirs, $\p H/\p C^j=0$, $j=1,...,s$.
This is because a non-empty kernel in $\omega$ 
results in ambiguity in the definition of $X$ (equation \eqref{CHM} is symmetric with respect to transformations $X\rightarrow X+Y$, with $Y\in {\rm ker}\lr{\omega}$).      
On the other hand, it should be noted that the proof of theorem 1 ultimately relies
on the proof of theorem 2 to find local canonical pairs $\lr{p^i,q^i}$, $i=1,...,m$, once 
the Casimir invariants have been identified (see \cite{Littlejohn82}).
The duality of the formulations is reflected in the properties of $\mc{J}$ and $\omega$. 
In particular, notice that the role of the Jacobi identity for the Poisson operator $\mc{J}$ 
is played by the closure of the symplectic $2$-form in the second formulation. 
This equivalence can be seen explicitly by considering the case in which the kernels of $\mc{J}$ and $\omega$ are empty, i.e. $n=2m$, $s=0$. 
Since in this situation the skew-symmetric matrix $\mc{J}^{ij}$ has full rank, it has an inverse. 
Comparing equations \eqref{HEoM} and \eqref{CHM} one sees that the inverse is given by the skew-symmetric matrix $\omega_{ij}$. 
Then, multiplying each side of equation \eqref{JI2} by $\omega_{iu}\omega_{jv}\omega_{kw}$ and summing over $i$, $j$, and $k$, one obtains:
\begin{equation}
\omega_{iu}\omega_{jv}\omega_{kw}\lr{\mc{J}^{im}\frac{\p\mc{J}^{jk}}{\p x^m}+\mc{J}^{jm}\frac{\p\mc{J}^{ki}}{\p x^m}+\mc{J}^{km}\frac{\p\mc{J}^{ij}}{\p x^m}}=\frac{\p\omega_{wv}}{\p x^u}+\frac{\p\omega_{uw}}{\p x^v}+\frac{\p\omega_{vu}}{\p x^w}=0,~~~~u,v,w=1,...,m,\label{closure2}
\end{equation}
which is the closure condition \eqref{closure}.
Conversely, multiplying each side of \eqref{closure2} by $\mc{J}^{iu}\mc{J}^{jv}\mc{J}^{kw}$ and summing over $u$, $v$, and $w$, one has:
\begin{equation}
\mc{J}^{iu}\mc{J}^{jv}\mc{J}^{kw}\lr{\frac{\p\omega_{wv}}{\p x^u}+\frac{\p\omega_{uw}}{\p x^v}+\frac{\p\omega_{vu}}{\p x^w}}=-\lr{\mc{J}^{iu}\frac{\p\mc{J}^{jk}}{\p x^u}+\mc{J}^{jv}\frac{\p\mc{J}^{ki}}{\p x^v}+\mc{J}^{kw}\frac{\p\mc{J}^{ij}}{\p x^w}}=0,~~~~i,j,k=1,...,m, 
\end{equation}
which is the Jacobi identity. Hence, the Jacobi identity and the closure of $\omega$ are mathematically equivalent. 
This fact will be the basic principle for the generalization of Hamiltonian mechanics presented in the following sections.  

\section{Generalized Hamiltonian Mechanics}
As anticipated in the introduction, the discussion will be limited to the generalization of Hamiltonian mechanics
to a three-dimensional phase space, i.e. $N=3$ and $n\geq 3$. Nevertheless, the same construction applies to $N>3$. 
We conjecture that a generalized Hamiltonian system should have the following properties. 
\begin{enumerate}
\item Given two Hamiltonians $G,H\in C^{\infty}\lr{\Omega}$, the equations of motion are given by
\begin{equation}
X^i=\mf{J}^{ijk}G_jH_k=\sum_{j<k}\mf{J}^{ijk}\lr{G_jH_k-G_kH_j},~~~~i=1,...,n.\label{GHEoM}
\end{equation}
Here, $\mf{J}^{ijk}\in C^{\infty}\lr{\Omega}$ are the components of a skew-symmetric third order contravariant tensor. 
The corresponding trivector field $\mf{J}\in\bigwedge^3 T\Omega$ is given by
\begin{equation}
\mf{J}=\sum_{i<j<k}\mf{J}^{ijk}\p_i\w\p_j\w\p_k.
\end{equation}
\item The trivector field $\mf{J}$ (generalized Poisson operator) satisfies a generalized Jacobi identity. 
This identity expresses the closure of a smooth $3$-form $w$ (symplectic $3$-form).
\item The symplectic $3$-form $w$ is Lie-invariant,
\begin{equation}
\mf{L}_Xw=0. 
\end{equation}
\item Given a generalized Poisson operator $\mf{J}$ with associated symplectic $3$-form $w$, 
there exists a local coordinate system $\lr{y^1,...,y^n}$ 
such that system \eqref{GHEoM} takes the generalized canonical form
\begin{equation}
\frac{dy^i}{dt}=\sum_{j<k}B^{ijk}\lr{\frac{\p G}{\p y^j}\frac{\p H}{\p y^k}-\frac{\p G}{\p y^k}\frac{\p H}{\p y^j}},~~~~i=1,...,n.\label{NC2}
\end{equation}
Here, $B^{ijk}\in\mathbb{R}$, $i,j,k=1,...,n$, are constant components of a skew-symmetric third order contravariant tensor $B$.  
Furthermore, the generalized Poisson operator and the symplectic $3$-form $w$ have local expressions
\begin{equation}
\mf{J}=\sum_{i<j<k}B^{ijk} \p_{y^i}\w\p_{y^j}\w\p_{y^k},~~~~w=\sum_{i<j<k}A_{ijk}dy^i\w dy^j\w dy^k,\label{Jwloc}
\end{equation}
where $\lr{\p_{y^1},...,\p_{y^n}}$ are tangent vectors and  
$A$ is the constant inverse of $B$ with skew-symmetric covariant components $A_{ijk}\in\mathbb{R}$, $i,j,k=1,...,n$ (the definition of inverse and the conditions of invertibility will be given later). 
When $n=3$, $B^{ijk}=b\epsilon^{ijk}$ for some constant $b\in\mathbb{R}$, and equation \eqref{NC2} gives Nambu's canonical equations \eqref{NC}. 
\item The local coordinate system $\lr{y^1,...,y^n}$ defines a phase space (Liouville) measure
\begin{equation}
d\Xi=dy^1\w...\w dy^n, 
\end{equation}
which is Lie invariant, i.e. 
\begin{equation}
\mf{L}_X d\Xi=0. 
\end{equation}
\end{enumerate}

In essence, the geometric construction outlined above can be thought of as a noncanonical formulation of Nambu mechanics.
For this construction to be consistent, we must show that the analogous of the Lie-Darboux theorem discussed in the previous
section applies to three-dimensional phase spaces. The proof of this theorem, which is given in section 5, 
can be obtained once certain subtleties concerning the notion of inverse and rank for tensors like $\mf{J}^{ijk}$ are settled. 
Then, the properties of the generalized formulation follow automatically.  
 
In the remainder of this section, we describe basic properties of the theory. 
First, notice that the generalized Poisson bracket (Nambu bracket) is defined as
\begin{equation}
\left\{F,G,H\right\}=\mf{J}^{ijk}F_iG_jH_k,~~~~\forall F,G,H\in C^{\infty}\lr{\Omega}.
\end{equation}
Observe that, while trilinearity, skew-symmetry, and Leibniz rule are satisfied by construction, 
the same is not true for the fundamental identity \eqref{FI}. This aspect will be discussed in detail in section 4.

A Casimir invariant $C$ is characterized by the property that it is a constant of the motion for any choice of Hamiltonians $G$ and $H$, i.e. 
\begin{equation}
\frac{dC}{dt}=\left\{C,G,H\right\}=0,~~~~\forall G,H\in C^{\infty}\lr{\Omega}.
\end{equation}  
In terms of the generalized Poisson operator, the equation above implies that the $1$-form $dC$ belongs to ${\rm ker}\lr{\mf{J}}$, 
$\mf{J}\lr{dC}=0$. In components, 
\begin{equation}
\mf{J}^{ijk}C_k=0,~~~~i,j=1,...,n.\label{Cas}
\end{equation}
At this point a remark on null spaces is useful. 
Since $\mf{J}$ is a trivector field, the following scenario may arise:
\begin{equation}
\mf{J}\lr{dC^1,dC^2}=0,~~~~\mf{J}\lr{dC^1}\neq0,~~~~\mf{J}\lr{dC^2}\neq 0.
\end{equation}
Neither $C^1$ or $C^2$ is a Casimir invariant, but they return $0$ when combined. 
For example, consider the following trivector field on $\mathbb{R}^6$,  
\begin{equation}
\mf{J}=\p_1\w\lr{\p_2\w\p_3+\p_4\w\p_5}.
\end{equation}
Clearly, $\mf{J}\lr{dx^6}=0$. However, we also have $\mf{J}\lr{dx^2-dx^4,dx^3+dx^5}=0$ with $\mf{J}\lr{dx^2-dx^4}=\p_1\w\lr{\p_5-\p_3}$ and $\mf{J}\lr{dx^3+dx^5}=\p_1\w\lr{\p_2+\p_4}$. $x^1$ behaves as a Casimir invariant, but $x^2-x^4$ and $x^3+x^5$ do not. We shall refer to quantities like $x^2-x^4$ and $x^3+x^5$ as semi-Casimir invariants. 
If one of the Hamiltonians, say $G$, happens to be a semi-Casimir invariant, 
it follows that the other semi-Casimir invariant is a constant of the motion independent of $H$. 

Suppose that there exists a local coordinate system $\lr{p^1,...,p^m,q^1,...,q^m,r^1,...,r^m,C^1,...,C^s}$ 
with $n=3m+s$ and tangent vectors $\lr{\p_{p^1},...,\p_{\p^m},\p_{q^i},...,\p_{q^m},\p_{r^1},...,\p_{r^m},\p_{C^1},...,\p_{C^s}}$   
such that 
\begin{equation}
\mf{J}=\sum_{i=1}^m\p_{p^i}\w \p_{q^i}\w \p_{r^i}.\label{Jcan}
\end{equation}
In these coordinates, the equations of motion take Nambu's canonical form
\begin{equation}
\frac{dp^i}{dt}=\frac{\p G}{\p q^i}\frac{\p H}{\p r^i}-\frac{\p G}{\p r^i}\frac{\p H}{\p q^i},~~~~
\frac{dq^i}{dt}=\frac{\p G}{\p r^i}\frac{\p H}{\p p^i}-\frac{\p G}{\p p^i}\frac{\p H}{\p r^i},~~~~
\frac{dr^i}{dt}=\frac{\p G}{\p p^i}\frac{\p H}{\p q^i}-\frac{\p G}{\p q^i}\frac{\p H}{\p p^i},~~~~\frac{dC^j}{dt}=0,
\end{equation}
with $i=1,...,m$ and $j=1,...,s$. 
Let $\tilde{\mf{J}}^{abc}$ denote the $\lr{a,b,c}$ component of $\mf{J}$ 
with respect to the coordinate system $\lr{p^1,...,p^m,q^1,...,q^m,r^1,...,r^m,C^1,...,C^s}$.
From equation \eqref{Jcan}, it follows that  
\begin{equation}
\tilde{\mf{J}}^{abc}=E^{abc},~~~~a,b,c=1,...,n,\label{GJc}
\end{equation}
where
\begin{equation}
E^{abc}=
\begin{cases}
\epsilon^{abc} & \text{if}~\sigma\lr{a,b,c}=\lr{i,m+i,2m+i},~~~~i=1,...,m\\
0 & \text{otherwise}.
\end{cases}\label{Jloc}
\end{equation}
In this notation, $\sigma$ is any rearrangement of the indexes $a,b,c$. 
In the next section it will be shown that skew-symmetric third order contravariant tensors with constant
entries satisfy the Jacobi identity. Therefore, \eqref{Jcan} is a generalized Poisson operator.   
Observe that when a generalized Poisson operator can be transformed in the form \eqref{Jcan}, 
its tensor representation can be written in terms of
Levi-Civita symbols. In particular, when there is only one canonical triplet ($m=1$), equation \eqref{Jloc} reduces to
\begin{equation}
\tilde{\mf{J}}^{abc}=\epsilon^{abc},~~~~a,b,c=1,2,3,
\end{equation}
which is the generalized Poisson operator originally introduced by Nambu.  
Let $\left\{\circ,\circ,\circ\right\}_c$ denote the Nambu bracket defined by a generalized Poisson operator
with components given by \eqref{GJc}. We shall refer to such bracket as canonical Nambu bracket.  
Given a generalized Poisson operator $\mf{J}$ with Nambu bracket $\left\{\circ,\circ,\circ\right\}$, 
the local coordinate system $\lr{p^1,...,p^m,q^1,...,q^m,r^1,...,r^m,C^1,...,C^s}$ transforms 
$\left\{\circ,\circ,\circ\right\}$ into $\left\{\circ,\circ,\circ\right\}_c$. 
Fixing one of the Hamiltonians, say $G$, define the bracket
\begin{equation}
\left\{\circ,\circ\right\}_G=\left\{\circ,G,\circ\right\}_c. 
\end{equation}
Let us show that, if $n=3$, $\left\{\circ,\circ\right\}_G$ is a Poisson bracket, and that
the same is not true, in general, for $n>3$. 
The candidate Poisson operator is 
\begin{equation}
\mc{J}_G=\sum_{a<c}\tilde{\mf{J}}^{abc}G_b\p_a\w\p_c=\sum_{a<c}E^{abc}G_b\p_a\w\p_c. 
\end{equation}
When $n=3$, we have two cases, $m=0$ and $m=1$. The case $m=0$ is trivial because the dimension of ${\rm ker}\lr{\mf{J}}$ is $s=n=3$, implying $\mf{J}=0$. If $m=1$, the expression above implies $\mc{J}^{ac}_G=\epsilon^{abc}G_b$. Then, the Jacobi identity \eqref{JI2} reads as 
\begin{equation}
\begin{split}
\mc{J}_G^{am}\frac{\p\mc{J}_{G}^{cd}}{\p x^m}&+\mc{J}_G^{cm}\frac{\p\mc{J}_{G}^{da}}{\p x^m}+\mc{J}_G^{dm}\frac{\p\mc{J}_{G}^{ac}}{\p x^m}\\=&\lr{\epsilon^{abm}\epsilon^{ced}+\epsilon^{cbm}\epsilon^{dea}+\epsilon^{dbm}\epsilon^{aec}}G_bG_{em}\\
=&\epsilon^{cad}\lr{\epsilon^{abm}G_{am}+\epsilon^{cbm}G_{cm}+\epsilon^{dbm}G_{dm}}G_b\\
=&\p_b\left[\epsilon^{cad}\lr{\epsilon^{abm}G_{am}+\epsilon^{cbm}G_{cm}+\epsilon^{dbm}G_{dm}}G\right]-4\epsilon^{cad}\lr{\epsilon^{abm}G_{amb}+\epsilon^{cbm}G_{cmb}+\epsilon^{dbm}G_{dmb}}G\\
=&\p_b\left[\epsilon^{cad}\lr{\epsilon^{abd}G_{ad}+\epsilon^{abc}G_{ac}+
\epsilon^{cba}G_{ca}+\epsilon^{cbd}G_{cd}+\epsilon^{dba}G_{da}+\epsilon^{dbc}G_{dc}}G\right]=0,~~~~a,c,d=1,2,3.
\end{split}
\end{equation}
Hence, for $n=3$, $\left\{\circ,\circ\right\}_G$ is a Poisson bracket. 
Next, suppose that $n>3$. Following similar steps as above with $\mc{J}^{ac}_G=E^{abs}G_b$, one obtains
\begin{equation}
\mc{J}_G^{am}\frac{\p\mc{J}_{G}^{cd}}{\p x^m}+\mc{J}_G^{cm}\frac{\p\mc{J}_{G}^{da}}{\p x^m}+\mc{J}_G^{dm}\frac{\p\mc{J}_{G}^{ac}}{\p x^m}=
\p_b\left[\lr{E^{amb}E^{cde}+E^{cmb}E^{dae}+E^{dmb}E^{ace}}GG_{em}\right],~~~~a,c,d=1,...,n.\label{JIE}
\end{equation}
Due to the skew-symmetry of $\mc{J}_G$, the indexes $a,c,d$ must be different from each other for the right-hand side
of the equation to be different from zero. In addition, 
at least one pair of these three indexes must belong to the same canonical triplet
due to the terms $E^{cde}$, $E^{dae}$, and $E^{ace}$.
If all three indexes belong to the same triplet, 
the corresponding term of the Jacobi identity identically vanishes as in the previous case $n=3$.
Hence, suppose that $a$ and $c,d$ belong to two different canonical triplets $m_1$ and $m_2$. Then, 
equation \eqref{JIE} reduces to
\begin{equation}
\mc{J}_G^{am}\frac{\p\mc{J}_{G}^{cd}}{\p x^m}+\mc{J}_G^{cm}\frac{\p\mc{J}_{G}^{da}}{\p x^m}+\mc{J}_G^{dm}\frac{\p\mc{J}_{G}^{ac}}{\p x^m}=
E^{amb}E^{cde}G_bG_{em}
=\lr{G_\beta G_{\epsilon\mu}-G_\mu G_{\epsilon\beta}},~~~~a,c,d=1,...,n,\label{JIE2}
\end{equation}
where $\lr{a,\mu,\beta}$ and $\lr{c,d,\epsilon}$ are the indexes of the canonical triplets $m_1$ and $m_2$. 
The quantity \eqref{JIE2} does not vanish in general. 
Notice that this result in agreement with the fact that the direct sum of 
canonical Nambu brackets with operator $\mf{J}=\p_1\w\p_2\w\p_3+\p_4\w\p_5\w\p_6+...+\p_{3L-2}\w\p_{3L-1}\w\p_{3L}$ on $\mathbb{R}^{3L}$ 
does not satisfy the fundamental identity \eqref{FI} when $L>1$ (see \cite{Tak}). 
This is because, as mentioned in the introduction, the fundamental identity \eqref{FI} 
for the Nambu bracket $\left\{\circ,\circ,\circ\right\}$ reduces to the Jacobi identity \eqref{JI} for the 
bracket $\left\{\circ,\circ\right\}_G$. 
It is worth observing that, however, the right-hand side of \eqref{JIE2} identically vanishes 
whenever the Hamiltonian $G$ is of the type
\begin{equation}
G=\sum_{i=1}^m G^i\lr{p^i,q^i,r^i,C^1,...,C^s},
\end{equation}
where the $G^i$ are $m$ functions depending only on the $i$th canonical triplet and the Casimir invariants. 

Finally, a remark on the quantization of the theory. 
The definition of Hamiltonian system provided by the present construction is weaker
than that resulting from the axioms \eqref{NAx}. In particular, the fundamental identity
is replaced by a weaker condition, the closure of the symplectic 3-form (see sections 4 and 5 for details). 
Therefore, we expect additional freedom in the derivation of a quantized bracket. 

\section{Symplectic 3-Form and Jacobi identity}
The purpose of the present section is to rigorously formulate the generalization of Hamiltonian mechanics introduced above 
in terms of a closed $3$-form, the symplectic $3$-form $w$, and to obtain the Jacobi identity for the generalized Poisson operator $\mf{J}$ by using the closure of $w$.
Consider again a smooth bounded domain $\Omega\subset\mathbb{R}^n$. 
The core of theory lies in the assumption that, given two Hamiltonians $G,H\in C^{\infty}\lr{\Omega}$, a vector field $X\in T\Omega$ 
defines a generalized Hamiltonian system provided that there exists a smooth $3$-form $w\in\bigwedge^3 T\Omega$ with the 
following properties,  
\begin{equation}
i_Xw=-dH\w dG,\label{w1}
\end{equation}
and
\begin{equation}
dw=0.\label{w2}
\end{equation}
The $3$-from $w$ can be expressed as
\begin{equation}
w=\sum_{i<j<k}w_{ijk}dx^i\w dx^j\w dx^k,\label{w3}
\end{equation}
where $w_{ijk}\in C^{\infty}\lr{\Omega}$, $i,j,k=1,...,n$, are the components of a skew-symmetric third order covariant tensor.  
Using \eqref{w3}, equations \eqref{w1} becomes
\begin{equation}
\sum_{j<k}X^iw_{ijk}dx^j\w dx^k=\sum_{j<k}\lr{H_kG_j-H_jG_k}dx^j\w dx^k,
\end{equation}
or,
\begin{equation}
X^iw_{ijk}=H_kG_j-H_jG_k,~~~~j,k=1,...,n.\label{Xw}
\end{equation}
Similarly, equation \eqref{w3} gives
\begin{equation}
dw=-\sum_{i<j<k<\ell}\lr{\frac{\p w_{ijk}}{\p x^\ell}+\frac{\p w_{i\ell j}}{\p x^k}+\frac{\p w_{ik\ell}}{\p x^j}+\frac{\p w_{j\ell k}}{\p x^i}}dx^i\w dx^j\w dx^k\w dx^\ell=0,
\end{equation}
or,
\begin{equation}
\frac{\p w_{ijk}}{\p x^\ell}+\frac{\p w_{i\ell j}}{\p x^k}+\frac{\p w_{ik\ell}}{\p x^j}+\frac{\p w_{j\ell k}}{\p x^i}=0,~~~~i,j,k,\ell=1,...,n.\label{dw0}
\end{equation}
At this point, a notion of inverse for third order tensors of the type $w_{ijk}$ is needed. 
Suppose that there exists a skew-symmetric third order contravariant tensor $\mf{J}$ with components $\mf{J}^{jk\ell}$ such that
\begin{equation}
\sum_{j<k}w_{ijk}\mf{J}^{jk\ell}=\delta_i^{\ell},~~~~i,\ell=1,...,n.\label{invw}
\end{equation}
Then, we say that $\mf{J}$ is the inverse of $w$. 
More generally, we say that $\mf{J}\in\bigwedge^3 T\Omega$ is a weak inverse of $w$ whenever the solution $X$ of the equation $i_Xw=-dH\w dG$ can be cast in the form $X^i=\mf{J}^{ijk}G_jH_k$ (the notion of weak invertibility will be discussed in detail in a subsequent publication). 
Multiplying both sides of equation \eqref{Xw} by $\mf{J}^{jk\ell}$, summing over repeated indexes, 
and using \eqref{invw}, one obtains
\begin{equation}
X^\ell=\sum_{j<k}\mf{J}^{jk\ell}\lr{H_kG_j-H_jG_k}=\mf{J}^{\ell jk}G_jH_k,~~~~\ell=1,...,n.
\end{equation}
This shows that system \eqref{w1} leads to the set of equations \eqref{GHEoM} if the inverse $\mf{J}$ exists. 
Let us derive necessary conditions for the existence of the inverse $\mf{J}$. 
First, observe that tensors like $w_{ijk}$ can be thought of as matrices that
have rows, columns, and `depth'. For example, the Levi-Civita symbol $\epsilon_{ijk}$ can be represented as
\begin{equation}
\begin{bmatrix} 0&0&0\\0&0&1\\0&-1&0\end{bmatrix}_{1},~~~~\begin{bmatrix}0&0&-1\\0&0&0\\1&0&0\end{bmatrix}_{2},~~~~\begin{bmatrix}0&1&0\\-1&0&0\\0&0&0\end{bmatrix}_{3}.
\end{equation}
Here, each matrix is numbered by the index $i$ (the depth of the matrix), while rows and columns by the indexes $j$ and $k$.
To the tensor $\epsilon_{ijk}$ we can also assign in a unique manner a conventional matrix having $n$ rows and $n^2$ columns as follows: 
\begin{equation}
\begin{bmatrix}
0&0&0&0&0&1&0&-1&0\\
0&0&-1&0&0&0&1&0&0\\
0&1&0&-1&0&0&0&0&0
\end{bmatrix}.
\end{equation}
Here, the rows are numbered by the index $i$, while columns by the pair $j,k$.   
The components of this $n\times n^2$ matrix will be denoted as $\epsilon_{i(jk)}$. 
Notice that the same construction applies to arbitrary tensors $A_{ijk}$, 
the associated $n\times n^2$ matrix being $A_{i\lr{jk}}$. 
Non-square matrices do not have a classical inverse. However, 
they may possess a left or right inverse \cite{LRinv}. 
In particular, given a $n\cp n^2$ matrix $A$ with components $A_{i\lr{jk}}$ of rank $n$,
the matrix $A$ has a right inverse, i.e. a matrix $B$ with components $B^{\lr{jk}\ell}$ such that $AB=2I_{n}$, 
where $I_n$ is the $n$-dimensional identity matrix \cite{LRinv} (the factor $2$ before $I_n $ is useful when handling skew-symmetric tensors).  
As an example, consider the matrix $\epsilon_{i\lr{jk}}$ above.
Evidently, its rank is $n=3$ because three of the columns are linearly independent vectors in $\mathbb{R}^{3}$. 
Hence, $\epsilon_{i\lr{jk}}$ has a right inverse. The components of the
right inverse are given by $B^{\lr{jk}\ell}=\epsilon^{\lr{jk}\ell}$ because
\begin{equation}
\epsilon_{ijk}\epsilon^{jk\ell}=2\sum_{j<k}\epsilon_{ijk}\epsilon^{jk\ell}=2\delta_{i}^{\ell},~~~~i,\ell=1,2,3. 
\end{equation} 
Therefore, the notion of invertibility \eqref{invw} for the tensor $w_{ijk}$ is related 
to the existence of a right inverse for the $n\times n^2$ matrix $w_{i\lr{jk}}$.  
The right inverse $\mf{J}^{\lr{jk}\ell}$ exists when $w_{i\lr{jk}}$ has rank $n$. 
In such case $w_{i\lr{jk}}$ is the left inverse of $\mf{J}^{\lr{jk}\ell}$. 

Next, we move to the Jacobi identity. 
In the present construction, the Jacobi identity is nothing but the closure condition \eqref{dw0} for the symplectic $3$-form $w$  
expressed in terms of the inverse $\mf{J}$. If the Jacobi identity is satisfied, we say that $\mf{J}$ is a generalized Poisson operator.  
Although there is no simple way to write \eqref{dw0} in terms of the components $\mf{J}^{ijk}$ 
(except directly substituting the expressions of the components $w_{ijk}$ as functions of $\mf{J}^{ijk}$),  
a necessary condition for equation \eqref{dw0} to hold where the components $w_{ijk}$ are partially removed can be derived as follows. 
Multiplying each side of \eqref{dw0} by $-\mf{J}^{\alpha ij}\mf{J}^{\beta k\ell}$ and summing over repeated indexes, we have:
\begin{equation}
\begin{split}
0=&-\mf{J}^{\alpha ij}\mf{J}^{\beta k\ell}\lr{\frac{\p w_{ijk}}{\p x^\ell}+\frac{\p w_{i\ell j}}{\p x^k}+\frac{\p w_{ik\ell}}{\p x^j}+\frac{\p w_{j\ell k}}{\p x^i}}\\
=&w_{ijk}\mf{J}^{\beta k\ell}\frac{\p\mf{J}^{\alpha ij}}{\p x^{\ell}}+w_{i\ell j}\mf{J}^{\beta k\ell}\frac{\p\mf{J}^{\alpha ij}}{\p x^{k}}+w_{ik\ell}\mf{J}^{\alpha ij}\frac{\p\mf{J}^{\beta k\ell}}{\p x^j}+w_{j\ell k}\mf{J}^{\alpha ij}\frac{\p\mf{J}^{\beta k\ell}}{\p x^i}\\
=&2w_{ijk}\mf{J}^{\beta k\ell}\frac{\p\mf{J}^{\alpha ij}}{\p x^\ell}+2w_{ik\ell}\mf{J}^{\alpha ij}\frac{\p\mf{J}^{\beta k\ell}}{\p x^j}\\
=&2w_{ijk}\lr{\mf{J}^{\beta k\ell}\frac{\p\mf{J}^{\alpha ij}}{\p x^\ell}+\mf{J}^{\alpha i\ell}\frac{\p\mf{J}^{\beta j k}}{\p x^\ell}}\\
=&2\sum_{i<j}w_{ijk}\lr{2\mf{J}^{\beta k\ell}\frac{\p\mf{J}^{\alpha ij}}{\p x^\ell}+\mf{J}^{\alpha i\ell}\frac{\p\mf{J}^{\beta j k}}{\p x^\ell}-\mf{J}^{\alpha j\ell}\frac{\p\mf{J}^{\beta i k}}{\p x^\ell}}\\
=&4\sum_{i<j<k}w_{ijk}\lr{
\mf{J}^{\beta k\ell}\frac{\p\mf{J}^{\alpha ij}}{\p x^\ell}+\mf{J}^{\alpha i\ell}\frac{\p\mf{J}^{\beta j k}}{\p x^\ell}+\mf{J}^{\alpha j\ell}\frac{\p\mf{J}^{\beta k i}}{\p x^\ell}
+\mf{J}^{\beta j\ell}\frac{\p\mf{J}^{\alpha ki}}{\p x^\ell}+\mf{J}^{\alpha k\ell}\frac{\p\mf{J}^{\beta i j}}{\p x^\ell}
+\mf{J}^{\beta i\ell}\frac{\p\mf{J}^{\alpha jk}}{\p x^\ell}
},\label{NCJ}
\end{split}
\end{equation}
where $\alpha,\beta=1,...,n$. Equation \eqref{NCJ} is a necessary condition that must be satisfied 
for the symplectic $3$-form $w$ to be closed. 
Observe that any invertible skew-symmetric third order tensor $\mf{J}^{ijk}$ with constant entries automatically satisfies \eqref{NCJ}, 
and equation \eqref{dw0} as well. 
Hence, the generalization of Hamiltonian mechanics following from the present construction is 
weaker than that resulting from enforcing the axioms \eqref{NAx}, since skew-symmetric third order tensors 
with constant entries do not satisfy the fundamental identity \eqref{FI} in general. 
Equation \eqref{NCJ} becomes a sufficient condition when $w$ and $\mf{J}$ satisfy the additional property
\begin{equation}
w_{ijk}\mf{J}^{k\ell \mu}=f\lr{\delta_i^\ell\delta_j^\mu-\delta_i^\mu\delta_j^\ell},~~~~i,j,\ell,\mu=1,...,n,~~~~f\in C^{\infty}\lr{\Omega}.\label{Inv2}
\end{equation}
Indeed, using \eqref{Inv2} one can recover the closure condition \eqref{dw0} by multiplying 
each side of \eqref{NCJ} by $w_{\sigma\tau\alpha}w_{\pi\rho\beta}$ and summing on repeated indexes. 
Notice that the property \eqref{Inv2} is satisfied by the Levi-Civita symbol. Indeed,
\begin{equation}
\epsilon_{ijk}\epsilon^{k\ell\mu}=\delta_i^\ell\delta_j^\mu-\delta_i^\mu\delta_j^\ell,~~~~i,j,\ell,\mu=1,...,3.
\end{equation}
When \eqref{Inv2} holds, the right hand side of equation \eqref{NCJ} reduces to
\begin{equation}
2f\lr{\delta_{i}^\ell \delta_j^\beta-\delta_{i}^\beta \delta_j^\ell}\frac{\p\mf{J}^{\alpha ij}}{\p x^{\ell}}+2f\lr{\delta_{j}^\ell \delta_k^\alpha-\delta_{j}^\alpha \delta_k^\ell}\frac{\p\mf{J}^{\beta jk}}{\p x^{\ell}}=4f\lr{\frac{\p\mf{J}^{\alpha \ell\beta}}{\p x^{\ell}}+\frac{\p\mf{J}^{\beta \ell\alpha}}{\p x^{\ell}}}=0,~~~~\alpha,\beta=1,...,n,
\end{equation}
This result implies that a smooth invertible $3$-form $w$ satisfying \eqref{Inv2} is closed
and therefore defines a symplectic $3$-form. 
We stress that, however, equation \eqref{Inv2} is only a sufficient condition for the closure of $w$, and it is not needed for $w$ to possess a right inverse. 

A final remark concerns how to identify 
a generalized Hamiltonian system from a given set of equations $X\in T\Omega$. 
While in classical Hamiltonian mechanics the Hamiltonian nature of a dynamical system
can be verified by checking the Poisson operator axioms, and, in particular, the Jacobi identity \eqref{JI2}, 
the procedure for the generalized case requires the determination of the symplectic $3$-form $w$.
Once the symplectic $3$-form $w$ and the Hamiltonians $G$ and $H$ such that \eqref{w1} holds
have been identified, one must check that $w$ is a closed differential form, i.e. that $dw=0$.
Notice that, in the presence of Casimir invariants, one may need to 
reduce the dynamics to a Casimir leaf in order to find the Hamiltonian structure. 



\section{Lie-Darboux and Liouville Theorems in Generalized Phase Space}

This section is dedicated to the proof of Lie-Darboux  type theorems (theorems 3 and 4) in the generalized Hamiltonian framework with 
a three-dimensional phase space $N=3$, $n\geq 3$. A direct consequence of these theorems is the local 
existence of an invariant (Liouville) measure. 
In particular, we prove a Lie-Darboux theorem (theorem 4) for closed $3$-forms of the type $w=\omega\w dG$, with $\omega$ a $2$-form and $dG$ an exact $1$-form.
A further result (proposition 1) is also proven explaining the relevance of this class of $3$-forms for generalized Hamiltonian mechanics, intended as the ideal dynamics of systems with 2 invariants. 

Below, we consider a smooth manifold $\Omega$ of dimension $n$ and assume smoothness of the involved quantities.
We have the following:

\begin{theorem}
Let $w\in\bigwedge^3T^{\ast}\Omega$ be a 
closed $3$-form. 
Let $w_{ijk}$, $i,j,k=1,...,n$ denote the components of $w$ with respect to a coordinate system $\lr{x^1,...,x^n}$ in  $\Omega$,
\begin{equation}
w=\sum_{i<j<k}w_{ijk}dx^i\w dx^j\w dx^k.
\end{equation}
Suppose that the $n\cp n^2$ matrix $w_{i\lr{jk}}$ has rank $n$.  
Take a sufficiently small neighborhood $U$ of any $\bol{x}_0\in\Omega$. 
Let $w_0=w_{0ijk}dy^i\w dy^j\w dy^k$ denote the constant (flat) $3$-form with components $w_{0ijk}=w_{ijk}\lr{\bol{x}_0}$ in a coordinate system $\lr{y^1,...,y^n}$.   
Further assume that 
Moser's $2$-form $\sigma_t$, $t\in\lrs{0,1}$, such that $d\sigma_t=d  w_t/d t$ in $U$, belongs to the image of the map $\hat{w}_t:T\Omega\rightarrow \bigwedge^2 T^{\ast}\Omega$ defined by $\hat{w}_t\lr{X_t}=-i_{X_t}w_t$, i.e. $\sigma_t\in {\rm Im}\lr{\hat{w}_t}$ for some $X_t\in T\Omega$.    
Then, 
$w_t$ has a right inverse $\mf{J}_t$ in $U$. Furthermore,  
there exists a 
coordinate change  $\lr{x^1,...,x^n}\rightarrow\lr{y^1,...,y^n}$
generated by the vector field $X_t=-\mf{J}^{jk\ell}_t\sigma_{tjk}\p_{\ell}$  
such that 
\begin{equation}
w=w_0~~~~{\rm in}~~U.\label{wc}
\end{equation}
\end{theorem}
\begin{proof}
We follow the steps of the classical proof of the Lie-Darboux theorem based on Moser's method \cite{DeLeon89,Mos1}.
Let $w_0$ denote the constant form on $\mathbb{R}^{n}$, 
\begin{equation}
w_0=\sum_{i<j<k}A_{ijk}dy^i\w dy^j\w dy^k,
\end{equation}
with $A_{ijk}$, $i,j,k=1,...,n$, real constants. 
Consider a family of vector fields $X_t\in T\Omega$, $0\leq t\leq 1$, defined in a neighborhood $U$ of a point $\bol{x}_0\in\Omega$ that generates
a one-parameter group of diffeomorphisms $g_t$ as follows, 
\begin{equation}
\frac{d}{dt}g_t\lr{\bol{x}_0}=X_t\lr{g_t\lr{\bol{x}_0}},~~~~g_0\lr{\bol{x}_0}=\bol{x}_0.
\end{equation}  
Next, define the family of $3$-forms
\begin{equation}
w_t=w_0+t\lr{w-w_0}.
\end{equation}
We wish to obtain $X_t$, and thus $g_t$, so that the following property is satisfied
\begin{equation}
g_t^\ast w_t=w_0.\label{pb}
\end{equation}
Here $g_t^\ast w_t$ denotes the pullback of $w_t$ by $g_t$. 
Equation \eqref{pb} implies that
\begin{equation}
\frac{d}{dt}g_t^\ast w_t=g_t^\ast\lr{\frac{dw_t}{dt}+di_{X_t}w_t}=0,\label{pb2} 
\end{equation}
where we used the fact that $w_t$ is a closed differential form. 
By the Poincar\'e lemma, in a sufficiently small neighborhood $W$ of $\bol{x}_0$, 
the closed differential form $dw_t/dt$ is exact, i.e. there exists a $2$-form $\sigma_t=\sum_{j<k}\sigma_{tjk}dx^j\w dx^k$ such that
\begin{equation}
\frac{dw_t}{dt}=d\sigma_t~~~~{\rm in}~~W.
\end{equation}
Hence, equation \eqref{pb2} can be solved in $W$ by finding a vector field $X_t$ satisfying
\begin{equation}
\sigma_t=-i_{X_t}w_t.\label{Xt1}
\end{equation} 
In components, equation \eqref{Xt1} is equivalent to
\begin{equation}
\sigma_{tjk}=-X_t^iw_{tijk},~~~~j,k=1,...,n.\label{Xt2}
\end{equation}
Next, observe that by hypothesis the $n\cp n^2$ matrix $w_{i\lr{jk}}$ has rank $n$. 
Similarly, setting the components of $w_0$ in the variables $\lr{y^1,...,y^n}$ 
to be given by the constant tensor $A_{ijk}=w_{ijk}\lr{\bol{x}_0}$,  
the $n\cp n^2$ matrix $A_{i\lr{jk}}$ has rank $n$.
Furthermore, at the point $\bol{x}_0$ we may assume $w\lr{\bol{x}_0}=w_0\lr{\bol{x}_0}$
because the matrices $w_{ijk}$ and $A_{ijk}$ coincide there. 
Then, for $0\leq t\leq 1$, 
\begin{equation}
w_t\lr{\bol{x}_0}=w_0\lr{\bol{x}_0}.
\end{equation}
This implies that the $n\cp n^2$ matrix $w_{ti\lr{jk}}\lr{\bol{x}_0}$ has rank $n$ at $\bol{x}_0$. 
By continuity of the tensor $w_{tijk}$ 
it follows that there exists a neighborhood $V$ of $\bol{x}_0$ where 
the rank of the $n\cp n^2$ matrix $w_{ti\lr{jk}}$ is $n$. 
Define $U=W\cap V$. Then, 
the matrix $w_{ti\lr{jk}}$ has a right-inverse inverse $\mf{J}_t^{\lr{jk}\ell}$. 
Since by hypothesis $\sigma_t\in{\rm Im}\lr{\hat{w}_t}$, 
in $U$ the solution $X_t$ of equation \eqref{Xt2} can be written in terms of the right inverse as
\begin{equation}
X_t^\ell=-\mf{J}_t^{jk\ell}\sigma_{tjk},~~~~\ell=1,...,n.\label{Xt3}
\end{equation}
The vector field \eqref{Xt3} gives the desired local change of coordinates. 
\end{proof}

\begin{theorem}
Let $\omega\in\bigwedge^2{T^{\ast}\Omega}$ denote a (not necessarily closed) $2$-form of constant rank $2m=n-s$ and $dG\in T{\Omega^{\ast}}$ an exact $1$-form such that $\lr{x^1,...,x^n}$ defines a coordinate system in $\Omega$ with $x^n=G$. Define the $3$-form $w\in\bigwedge^3{T^{\ast}\Omega}$ as $w=\omega\w dG$ 
and suppose that $dw=0$. Then, for every $\bol{x}_0\in \Omega$ there exist a neighborhood $U$ of $\bol{x}_0$ and a coordinate system $\lr{p^1,...,p^\ell,q^1,...,q^\ell,G^1,...,G^\tau}$ with $n=2\ell+\tau$ such that 
\begin{equation}
w=\omega_0\w dG,~~~~\omega_0=\sum_{i=1}^\ell dp^i\w dq^i~~~~{\rm in}~~U,
\end{equation}
with $\ell=m$ if $\p_n\in{\rm ker}\lr{\omega}$ and $2\ell\leq n-1$ if $\p_n\notin {\rm ker}\lr{\omega}$. 
Furthermore, given a $1$-form $dH\in {T^{\ast}\Omega}$, linearly independent from $dG$, the phase space measure $d\Pi=dp^1\w ...\w dp^{\ell}\w dq^1\w ...\w dq^{\ell}\w dG^1\w...\w dG^\tau$ is an invariant measure in $U$ for the generalized Hamiltonian system  $X\in T{U}$ such that
\begin{equation}
i_Xw=-dH\w dG,\label{sys}
\end{equation}
provided that such $X$ exists. In addition,
\begin{equation}
i_X\omega_0=-\tilde{d}H~~~~{\rm in}~~\Sigma_{G},
\end{equation}
where $\Sigma_{G}=\lrc{\bol{x}\in U:G\lr{\bol{x}}=c\in\mathbb{R}}$ and $\tilde{d}$ denotes the differential operator on $\Sigma_{G}$. 
\end{theorem}

\begin{proof}
Since $dw=d\omega\w dG=0$, it follows that $\tilde{d}\omega=0$ in any level set $\Sigma_{G}$. 
On the other hand, 
\begin{equation}
\omega=\sum_{i<j}\omega_{ij}dx^{i}\w dx^j=\sum_{i=1}^{n-1}\omega_{in}dx^i\w dG+\sum_{i<j}^{n-1}\omega_{ij}dx^{i}\w dx^j.\label{omega}
\end{equation}
Define $\tilde{\omega}=\sum_{i<j}^{n-1}\omega_{ij}dx^{i}\w dx^j$. Evidently $w=\tilde{\omega}\w dG$.  Since $w$ is closed, this implies $\tilde{d}\tilde{\omega}=0$. If $\p_n\in{\rm ker}\lr{\omega}$, from \eqref{omega} it follows that $\omega=\tilde{\omega}$ and ${\rm rank}\lr{\tilde{\omega}}=2m=n-s$. Conversely,   
if $\p_n\notin{\rm ker}\lr{\omega}$ the forms $\omega$ and  $\tilde{\omega}$ are different, with ${\rm rank}\lr{\tilde{\omega}}=2\ell=n-1-u\leq n-1$. 
In either case, by the Lie-Darboux theorem for all $\bol{x}_0\in\Omega$ there exist a neighborhood $U$ of $\bol{x}_0$ and $n-1$ local coordinates $\lr{p^1,...,p^\ell,q^1,...,q^\ell,G^1,...,G^{s-1}}$ or $\lr{p^1,...,p^\ell,q^1,...,q^\ell,G^1,...,G^{u}}$  
such that 
\begin{equation}
\tilde{\omega}=\omega_0=\sum_{i=1}^{\ell}\tilde{d}p^i\w \tilde{d}q^i~~~~{\rm in}~~\Sigma_{G}. 
\end{equation}
By smoothness, the coordinates  $p^i,q^i:C^{\infty}\lr{\Sigma_{G}}\rightarrow\mathbb{R}$ also define smooth functions 
$p^i,q^i:C^{\infty}\lr{U}\rightarrow\mathbb{R}$. Then,
\begin{equation}
w=\sum_{i=1}^{\ell}\tilde{d}p^i\w\tilde{d}q^i\w dG=
\sum_{i=1}^{\ell}{d}p^i\w{d}q^i\w dG=\omega_0\w dG.
\end{equation}
Now consider a solution $X\in TU$ of system \eqref{sys}. Recalling that, by hypothesis, $dH$ and $dG$ are linearly independent and noting that $0=i_Xi_Xw=-\lr{i_XdH} dG+\lr{i_XdG} dH$, it follows that $i_XdH=i_XdG=0$. On the other hand, $i_Xw=i_X\omega_0\w dG=-\tilde{d}H\w dG$, which implies
\begin{equation}
i_X\omega_0=-\tilde{d}H~~~~{\rm in}~~\Sigma_{G}.\label{sys2}
\end{equation}
Since $\tilde{d}\tilde{\omega}=0$, equation \eqref{sys2} defines a Hamiltonian system with invariant measure $\lr{\bigwedge_{i=1}^{\ell}\tilde{d}p^i\w \tilde{d}q^i}\w \tilde{d}G^1\w ...\w \tilde{d}G^{s-1}$ if $\p_n\in{\rm ker}\lr{\omega}$ or 
$\lr{\bigwedge_{i=1}^{\ell}\tilde{d}p^i\w \tilde{d}q^i}\w \tilde{d}G^1\w ...\w \tilde{d}G^{u}$
if $\p_n\notin{\rm ker}\lr{\omega}$ on $\Sigma_{G}$.  
Set $\lr{G^1,...,G^{\tau}}=\lr{G^1,...,G^{s-1},G}$ if $\p_n\in{\rm ker}\lr{\omega}$ and $\lr{G^1,...,G^{\tau}}=\lr{G^1,...,G^{u},G}$ if $\p_n\notin{\rm ker}\lr{\omega}$. It follows that
\begin{equation}
d\Pi=\lr{\bigwedge_{i=1}^{\ell}dp^i\w dq^i}\w dG^1\w...\w dG^\tau,
\end{equation}
defines an invariant measure for $X$ in $U$. 
\end{proof}

\begin{proposition}
Let $w\in\bigwedge^3T^{\ast}\Omega$ denote a closed 3-form and $dG\in  T^{\ast}\Omega$ an exact $1$-form such that $\lr{x^1,...,x^{n}}$ defines a coordinate system in $\Omega$ with $x^n=G$.  
Suppose that for any exact 1-form $dH\in T^{\ast}\Omega$ such that $dH$ and $dG$ are linearly independent there exists a vector field $X\in T\Omega$ solving 
\begin{equation}
i_{X}w=-dH\w dG.\label{EoM33}
\end{equation} 
Further assume 
 that the $2$-tensor $\omega_{ij}=w_{ijn}$ is invertible on the level sets $\Sigma_{G}=\lrc{\bol{x}\in\Omega:G\lr{\bol{x}}=c\in\mathbb{R}}$ with inverse $\mc{J}\in\bigwedge^2 T\Sigma_G$ such that
\begin{equation}
\sum_{j=1}^{n-1}\omega_{ij}\mc{J}^{jk}=\delta_i^k,~~~~i,k=1,...,n-1.\label{invo3}
\end{equation}
Then, on each level set $\Sigma_{G}$ there exists a closed 2-form $\tilde{\omega}\in\bigwedge^2 T^{\ast}\Sigma_G$ such that
\begin{equation}
i_{X}\tilde{\omega}=-\tilde{d}H,
\end{equation}
where $\tilde{d}$ denotes the differential operator on $\Sigma_{G}$. Furthermore, 
\begin{equation}
w=\tilde{\omega}\w dG, 
\end{equation}
and
\begin{equation}
X=\mf{J}\lr{dH,dG},
\end{equation}
with $\mf{J}=\mc{J}\w \p_n$. 
\end{proposition}

\begin{proof}
Equation \eqref{EoM33} implies that
\begin{equation}
X^iw_{ijk}=H_kG_j-H_jG_k.
\end{equation}
Since $x^n=G$, we have $X^iw_{ijn}=-H_j$ for $j=1,...,n-1$. 
Hence, 
\begin{equation}
i_X\omega=-\tilde{d}H,\label{EoM44}
\end{equation}
where $\omega\in\bigwedge^2T^{\ast}\Omega$ is the 2-form $\omega=\sum_{i<j}\omega_{ij}dx^i\w dx^j$ and $\tilde{d}$ is the differential operator on the level sets $\Sigma_G$. 
Since $i_XdG=0$, the equations of motion \eqref{EoM33} and \eqref{EoM44} give 
\begin{equation}
i_X\lr{w-\omega\w dG}=0.\label{EoM55}
\end{equation}
Let $\xi\in\bigwedge^3 T^{\ast}\Omega$  denote a 3-form such that 
\begin{equation}
i_X\xi=\sum_{j<k}X^i\xi_{ijk}dx^j\w dx^k=\sum_{j<k}\mc{J}^{i\ell}H_{\ell}\xi_{ijk}dx^j\w dx^k=0.
\end{equation}
It follows that
\begin{equation}
w-\omega\w dG=\xi.
\end{equation}
On the other hand $w$, and thus $\xi$, cannot depend on $H$ by construction. Therefore, we must have
\begin{equation}
\mc{J}^{i\ell}\xi_{ijk}=0~~~~\forall~~\ell=1,...,n-1,~~j,k=1,...,n.
\end{equation}
However, the tensor $\mc{J}$ is invertible on $\Sigma_G$ by hypothesis (equation \eqref{invo3}). Hence, $\xi=0$ must be the zero 3-form. 
Then, equation \eqref{EoM55} can be expressed in the form
\begin{equation}
w=\omega\w dG.
\end{equation}
Using the closure of $w$, we therefore arrive at the equation
\begin{equation}
0=d\omega\w dG=\tilde{d}\omega\w dG.\label{dodG}
\end{equation}
However, the 3-form  
$\tilde{d}\omega$ can be expanded on the basis elements $dx^i\w dx^j\w dx^k$ with $i<j<k$ and $i,j,k=1,...,n-1$, which satisfy $dx^i\w dx^j \w dx^k\w dG\neq 0$. It follows that
\begin{equation}
\tilde{d}{\omega}=0,
\end{equation}
i.e. the 2-form ${\omega}\in\bigwedge^2T^{\ast}\Sigma_G$ is closed. The theorem is proven by noting that $X=\mf{J}\lr{dH,dG}$ with $\mf{J}=\mc{J}\w \p_n$ and by setting $\tilde{\omega}=\omega$. 
\end{proof}
We remark that proposition 1 applies to the case in which $n$ is odd, because the invertibility of $\omega$ implies that $n=2m+1$ for some $m\in\mathbb{N}$. The case in which $n$ is even can be handled by a further integrability assumption on the kernel of $\omega$.

We conclude this section with some observation concerning invertible $3$-forms $w$ that admit a constant (flat) expression  $w_0=\sum_{i<j<k}A_{ijk}dy^i\w dy^j\w dy^k$, $A_{ijk}\in\mathbb{R}$, by a suitable change of coordinates.    
First, any 
such form induces 
an invariant (Liouville) measure given by the phase space volume
\begin{equation}
d\Xi=dy^1\w ...\w dy^n.
\end{equation}
To see this observe that  
the equations of motion \eqref{w1} take the local form
\begin{equation}
A_{ijk}\frac{dy^i}{dt}=\frac{\p G}{\p y^j}\frac{\p H}{\p y^k}-\frac{\p G}{\p y^k}\frac{\p H}{\p y^j},~~~~j,k=1,...,n.\label{locEq1}
\end{equation}
The existence of the inverse $B^{jk\ell}$ 
of $A_{ijk}$ implies that 
\begin{equation}
\frac{dy^\ell}{dt}=B^{\ell jk}\frac{\p G}{\p y^j}\frac{\p H}{\p y^k},~~~~\ell=1,...,n.\label{locEq2}
\end{equation}
It follows that,
\begin{equation}
\mf{L}_Xd\Xi=\frac{\p}{\p y^i}\lr{\frac{dy^i}{dt}}d\Xi=B^{ijk}\lr{\frac{\p^2 G}{\p y^i\p y^j}\frac{\p H}{\p y^k}+\frac{\p G}{\p y^k}\frac{\p^2 H}{\p y^i\p y^j}}d\Xi=0.
\end{equation}
Notice that, in addition to $d\Xi$, the closure condition implies that the symplectic $3$-form $w$ is Lie-invariant as well, 
\begin{equation}
\mf{L}_Xw=di_Xw=-d\lr{dH\w dG}=0.
\end{equation}
Therefore, in the generalized Hamiltonian framework developed here 
both the symplectic form $w$ and the Liouville measure $d\Xi$ 
are preserved as in the classical formulation. 
However, a difference exists with respect to the existence of canonical variables. 
Indeed, while in the classical proof of the Lie-Darboux theorem 
the skew symmetry of the tensor $\omega_{ij}$ associated with the symplectic $2$-form $\omega$ is sufficient to 
ensure that there exists a linear change of basis transforming the skew-symmetric matrix $\omega_{ij}\lr{\bol{x}_0}$
into block diagonal form at any $\bol{x}_0\in\Omega$, an analogous result is not available for third order tensors like $w_{ijk}\lr{\bol{x}_0}$.
For this reason, one cannot guarantee the local invertibility of the tensor $\tilde{w}_{tijk}$ associated with the $3$-form $\tilde{w}_t=\tilde{w}_0+t\lr{w-\tilde{w}_0}$, with
\begin{equation}
\tilde{w}_0=\sum_{i=1}^m dp^i\w dq^i\w dr^i.
\end{equation}  
Hence, the proof of theorem 3 breaks down, i.e. local canonical triplets $\lr{p^i,q^i,r^i}$, $i=1,...,m$, are not available in general. 
Furthermore, even if $n=3m$ with $m$ an integer, for canonical triplets of variables 
$\lr{p^1,...,p^m,q^1,...,q^m,r^1,...,r^m}$ to  
locally exist in the neighborhood of all points $\bol{x}_0\in\Omega$, it is not sufficient that 
$w_{ijk}\lr{\bol{x}_0}$ can be transformed by a linear change of basis into the generalized Levi-Civita symbol $E_{ijk}$ (the covariant version of the 
tensor \eqref{Jloc} introduced in section 3), because the applicability of theorem 3 also requires that the relevant Moser $2$-form $\tilde{\sigma}_t$ belongs to the image of the map $\hat{\tilde{w}}_t$. 

\section{Concluding Remarks}

In this paper, we have formulated a generalization of classical Hamiltonian mechanics 
to a three-dimensional phase space. 
The generalized theory relies on a symplectic $3$-form $w$ and a pair of Hamiltonian functions $G,H$.  
The closure of the symplectic $3$-form $w$ ensures that there exist a local coordinate system $\lr{y^1,...,y^n}$
such that the components of $w$ are constants, and the volume form $d\Xi=d y^1\w ...\w dy^n$ is Lie-invariant
with respect to the generalized Hamiltonian flow. 
The invariant volume element can be used to define the phase space measure of statistical mechanics. 
When the components of $w$ define an $n\cp n^2$ matrix of rank $n$, the form $w$ has a right inverse. 
If the right-inverse corresponds to a an antisymmetric $3$-tensor,  
it defines a generalized Poisson operator $\mf{J}$. 
In analogy with the classical theory, 
the Jacobi identity is identified with the closure of the symplectic $3$-form $w$ 
expressed in terms of $\mf{J}$. As a consequence, the Jacobi identity
is a weaker condition than the fundamental identity, and any 
skew-symmetric third order contravariant tensor with constant entries defines a generalized Poisson operator. 
Skew-symmetric third order tensors also exhibit a richer kernel structure. 
In particular, semi-Casimir invariants produce conservation laws when
they appear together in the generalized Poisson bracket. 
 
There are aspects of the theory that deserve further investigation. 
On one hand it is desirable to identify the mathematical conditions
for the existence of a linear change of basis 
transforming the constant skew-symmetric third order covariant tensor $w_{ijk}\lr{\bol{x}_0}$, $\bol{x}_0\in\Omega$,
into the generalized Levi-Civita symbol $E^{ijk}$. 
Such conditions will provide necessary conditions for the
local existence of canonical triplets $\lr{p^i,q^i,r^i}$, $i=1,...,m$. 
On the other hand, it would be useful to identify
systems that are Hamiltonian in the generalized sense,
but that do not possess a classical Hamiltonian structure.
Such systems would give the theory a physical foundation. 
Finally, the quantization of the generalized Poisson bracket
resulting from the present construction has not been investigated.

\section*{Acknowledgment}
The research of NS was partially supported by JSPS KAKENHI Grant No. 17H01177 
and by the Osaka City University Advanced Mathematical Institute (MEXT Joint Usage/Research Center on Mathematics and Theoretical Physics JPMXP0619217849).

N.S. gratefully acknowledges N. Duignan (University of Sydney) for  valuable insight in helping identify the issue with the formulation of Theorem 3 
in the previous version of this manuscript, and for providing useful suggestions on the present version.

\section*{Data Availability}

The data that support the findings of this study are available from
the corresponding author upon reasonable request.

\end{document}